\documentclass[a4paper,UKenglish,cleveref, autoref, thm-restate]{lipics-v2021}
\usepackage{todonotes}

\title{Parameterized complexity of reconfiguration of atoms} 

\titlerunning{Parameterized complexity of reconfiguration of atoms} 

\author{Alexandre Cooper}{Institute of Quantum Computing, University of Waterloo, Canada}{alexander.cooper@uwaterloo.ca}{[orcid]}{This research was undertaken thanks in part to funding from the Canada First Research Excellence Fund.}

\author{Stephanie Maaz}{Department of Computer Science, American University of Beirut, Lebanon}{sjm14@mail.aub.edu}{[orcid]}{Supported by URB project ``A theory of change through the lens of reconfiguration''}

\author{Amer~E.~Mouawad}{Department of Computer Science, American University of Beirut, Lebanon}{amer.mouawad@gmail.com}{[orcid]}{Supported by URB project ``A theory of change through the lens of reconfiguration''}

\author{Naomi Nishimura}{David R. Cheriton School of Computer Science, Univeristy of Waterloo, Canada}{nishi@uwaterloo.ca}{[orcid]}{Research supported by the Natural Sciences and Engineering Research Council of Canada}

\authorrunning{A.\,Cooper, S. Maaz, A.~E.~Mouawad, and N.\, Nishimura} 

\Copyright{John Q. Public and Joan R. Public} 



\category{} 

\relatedversion{} 

\nolinenumbers 

\hideLIPIcs  

\EventEditors{John Q. Open and Joan R. Access}
\EventNoEds{2}
\EventLongTitle{42nd Conference on Very Important Topics (CVIT 2016)}
\EventShortTitle{CVIT 2016}
\EventAcronym{CVIT}
\EventYear{2016}
\EventDate{December 24--27, 2016}
\EventLocation{Little Whinging, United Kingdom}
\EventLogo{}
\SeriesVolume{42}
\ArticleNo{23}

\begin{document}
\maketitle
\begin{abstract}
Our work is motivated by the challenges presented in preparing arrays of atoms for use in quantum simulation.  The recently-developed process of loading atoms into traps results in approximately half of the traps being filled.  To consolidate the atoms so that they form a dense and regular arrangement, such as all locations in a grid, atoms are rearranged using moving optical tweezers.  Time is of the essence, as the longer that the process takes and the more that atoms are moved, the higher the chance that atoms will be lost in the process.

Viewed as a problem on graphs, we wish to solve the problem of reconfiguring one arrangement of tokens (representing atoms) to another using as few moves as possible.  
Because the problem is NP-complete on general graphs as well as on grids~\cite{DBLP:journals/siamdm/CalinescuDP08}, we focus on the 
parameterized complexity for various parameters, considering both undirected and directed graphs, and tokens with and without labels. 
For unlabelled tokens, the problem is in FPT when parameterizing by the number of tokens, the number of moves, or the number of moves plus the number 
of vertices without tokens in either the source or target configuration, but intractable when parameterizing by the difference between the number of moves and the number of differences in the 
placement of tokens in the source and target configurations. When labels are added to tokens, however, most of the tractability results are replaced by hardness results.
\end{abstract}

\section{Introduction}\label{sec-intro}

To maximize the probability of success in arranging atoms, approaches need to minimize the probability of atoms being lost during the time between the array being loaded and the atoms being arranged.   The lifetime of trapped atoms is short and limited, and the process of moving an atom may result in the loss of the atom.  Previous work~\cite{PhysRevA.102.063107} has focused on minimizing the total time required, including both the generation and the execution of the sequence of steps, and consequently has aimed to minimize the number of moves.

The rearrangement of atoms can be framed as a reconfiguration problem;
the \emph{reconfiguration framework}~\cite{IDHPSUU11, H13, DBLP:journals/algorithms/Nishimura18, CHJ08,IDHPSUU11,IKD12,DBLP:journals/talg/LokshtanovM19} characterizes the transformation between \emph{configurations} by means of a sequence of \emph{reconfiguration steps}.  By representing atoms as tokens, we can define each configuration of unlabelled tokens as a subset of vertices of a graph,  indicating that there is a token placed on each vertex in the subset; for tokens with labels,  a \emph{labelled configuration} consists of a sequence of vertices in a graph, where the position of a vertex in the sequence corresponds to its label. One configuration can be transformed into another by a sequence of moves, where in each move a token is  moved from one vertex to another along a token-free path.

Since finding a shortest sequence of moves between configurations is NP-hard, even when restricted to grids~\cite{DBLP:journals/siamdm/CalinescuDP08}, we turn to 
the field of parameterized complexity~\cite{DF97,flumgrohe,niedermeier2006,saurabh-book}, which studies the impact of one or more \emph{parameters} on the running time of algorithms.  
A problem is in \emph{FPT} if there exists an algorithm with worst-case running time bounded by $f(k) \cdot n^{O(1)}$ for $n$ the size of the instance, 
$k$ the size of the parameters, and $f$ a computable function; analogous to NP-hardness in the realm of classical complexity are the classes of intractable problems in the \emph{W-hierarchy} such as W[1] and W[2].

We explore the fixed-parameter tractability of the problem for unlabelled and labelled tokens on undirected and directed graphs, with 
respect to various parameters, namely, the number of tokens ($k$), the number of moves ($\ell$), the number of token-less vertices outside the source and target configurations ($f$), and the number of moves exceeding the minimum possible for any instance (namely, the number of differences between the source and target configurations). 
Our results are summarized in Table~\ref{table-summary}.

\begin{table}[h]
\begin{centering}
  \begin{tabular}{|l||l|l|l|l|l|l|l|}
    \hline
  & $k$ & $\ell$ & $\ell + f$ & $\ell - |S \setminus T|$ \cr
  \hline
\hline  
\textsc{UUTM} & FPT (Cor~\ref{cor-kernel}) & FPT (Thm~\ref{thm-fpt-ell}) & FPT (Thm~\ref{thm-fpt-ell}) & W[2]-hard  (Thm~\ref{thm-broom})\cr
\hline
\textsc{UDTM} & FPT (Cor~\ref{cor-kernel}) & FPT (Thm~\ref{thm-dir-ell}) & FPT (Thm~\ref{thm-dir-ell}) & W[2]-hard  (Thm~\ref{thm-broom})\cr
\hline
\textsc{LUTM} & Open & W[1]-hard (Thm~\ref{thm-undir-free}) & W[1]-hard (Thm~\ref{thm-undir-free}) & W[2]-hard  (Thm~\ref{thm-broom})\cr
  \hline
\textsc{LDTM} & Open & W[1]-hard (Thm~\ref{thm-dir-free}) & W[1]-hard (Thm~\ref{thm-dir-free}) & W[2]-hard  (Thm~\ref{thm-broom})\cr
\hline
  \end{tabular}
  \caption{Summary of results for \textbf{U}nlabelled/\textbf{L}abelled and \textbf{U}ndirected/\textbf{D}irected problem variants}
  \label{table-summary}
\end{centering}
\end{table}

\section{Terminology}\label{sec-term}
We formulate our problems in terms of the moving of tokens in a graph, using the notation $G = (V(G),E(G))$ for an undirected 
graph and $D = (V(D),E(D))$ for a directed graph.  The reader is directed to a standard textbook on graph theory~\cite{DBLP:books/daglib/0030488} for 
definitions of graph classes and other terminology.

We define a \emph{move} as a pair $(s,t)$, where $s$ is the \emph{source vertex} of the move and $t$ is the \emph{target vertex} of the move.  
The \emph{execution} of the move $(s,t)$ results in the change from a configuration containing $s$ to a configuration containing $t$, where the same label is mapped to $s$ and $t$, with the rest of the configuration remaining unchanged.  
  In order to ensure that atoms do not collide, a move cannot pass through a vertex that contains a token.  
A vertex is \emph{free} if there is no token on it, and  a path (directed or undirected) is \emph{free} if all intermediate vertices in the path are free. 
A move $(s,t)$ in a sequence is \emph{valid} if, after the execution of the previous moves in the sequence, there is a token on $s$, $t$ is free, and there is 
a free path from $s$ to $t$, which we designate as the \emph{path for the move}.  For a sequence of valid moves $\alpha$ in a graph $G$ (respectively, $D$), we use $G_{\alpha}$ ($D_{\alpha}$) to 
denote the graph induced on the union of edges in the paths for the moves in $\alpha$.  

The execution of a sequence of valid moves \emph{transforms} a configuration $S$ into another configuration $T$ if executing 
the moves starting from $S$ results in tokens being placed as in configuration $T$.   We will call such a sequence of valid moves 
a \emph{transforming sequence for $S$ and $T$} or, when $S$ and $T$ are clear from context, simply a \emph{transforming sequence}.  
In defining a sequence of indices, we use $[n]$ to denote $\{1, \ldots, n\}$. 

We define the four problems \textsc{Unlabelled Undirected Token Moving (UUTM)},
\textsc{Unlabelled Directed Token Moving (UDTM)}, \textsc{Labelled Undirected Token Moving (LUTM)}, and \textsc{Labelled Directed Token Moving (LDTM)} as follows:

\begin{description}
\item[Input (\textsc{UUTM})] An undirected graph $G$ and configurations $S \subseteq V(G)$ and $T \subseteq V(G)$ such that $|S| = |T| = k$ and an integer $\ell$
\item[Input (\textsc{UDTM})] A digraph $D$ and configurations $S \subseteq V(D)$ and $T \subseteq V(D)$ such that $|S| = |T| = k$ and an integer $\ell$
\item[Input (\textsc{LUTM})] An undirected graph $G$ and labelled configurations $S \subseteq V(G)$ and $T \subseteq V(G)$ such that $|S| = |T| = k$ and an integer $\ell$
\item[Input (\textsc{LDTM})] A digraph $D$ and labelled configurations $S \subseteq V(D)$ and $T \subseteq V(D)$ such that $|S| = |T| = k$ and an integer $\ell$  
\item[Output (all problems)] A transforming sequence of length at most $\ell$ 
\end{description}





Unless specified otherwise, all definitions apply to instances $(G, S, T, \ell)$ of \textsc{UUTM} and \textsc{LUTM} as well as 
instances $(D, S, T, \ell)$ of \textsc{UDTM} and \textsc{LDTM}. 
We refer to $S$ as the \emph{source configuration}, $T$ as the \emph{target configuration}, $O = S \cap T$ as the 
set of \emph{obstacles}, and $S \Delta T$ as the \emph{symmetric difference} of $S$ and $T$.
In a \emph{clearing move}, the source vertex is an obstacle and in a \emph{filling move}, the target vertex is an obstacle.

When discussing parameters, we use $k$ to denote $|S| = |T|$ and
$f$ to denote $|V(G)| - |S \cup T|$ or $|V(D)| - |S \cup T|$ (the number of vertices without tokens in either $S$ or $T$).

We will refer to an instance $(G, S, T)$ (respectively, $(D, S, T)$) when discussing the length of a shortest transforming sequence from $S$ to $T$ in $G$ ($D$).  
Two instances are \emph{equivalent} if the lengths of the shortest transforming sequences of the instances are equal.

For a (directed) path between a pair of vertices $s$ and $t$, any token on a vertex other than $s$ or $t$ is said to \emph{block} that path.  
If all paths between a pair of vertices $s$ and $t$ are blocked, then we say that the move $(s,t)$ is \emph{blocked}.  
For shorthand, when the presence of a token on a vertex $v$ results in a move $m$ being blocked, we'll say that $m$ \emph{is blocked by} $v$.  

The observations below follow from the definitions:

\begin{observation}\label{obs-block-path}
  If a vertex $v$ blocks move $(s,t)$ in $G$ (respectively, $D$), then there exists a path (resp., directed path) from $v$ to $t$.
\end{observation}

\begin{observation}\label{obs-remove-last}
Suppose that $m = (s,t)$ is the last move in a transforming sequence $\alpha$ from $S$ to $T$.  The sequence $\alpha'$ formed from $\alpha$ by removing $m$ is a transforming sequence from $S$ to a configuration $T'$, where $T'$ differs from $T$ by having a token on $s$ instead of on $t$.  
\end{observation}

Taken together, we can use the observations to detect the existence of free paths:

\begin{observation}\label{obs-only-block}
For $\alpha$, $\alpha'$, $m = (s,t)$, $S$, and $T$ as in Observation~\ref{obs-remove-last}, suppose there exists a transforming sequence $\gamma$ from $S$ to a configuration $U$, where $U$ differs from $T'$ by a single token, where $U$ has a token on a vertex $u$ and $T'$ has a token on a vertex $v \ne s$.  Then if $m$ is blocked after the execution of $\gamma$ on $S$, there must be a free (directed, if in $D$) path from $u$ to $t$.
\end{observation}  

\begin{proof}
We know that $m$ is valid in $\alpha$ and hence $m$ is valid after the execution of $\alpha'$.  The only difference between $U$ and $T'$ is the presence of a token on $u$ instead of $v$, so the only vertex that can block $m$ is $u$.  The fact that $u$ is the only reason that there is not a free path from $s$ to $t$, combined with Observation~\ref{obs-block-path}, implies that there is a free (directed) path from  $u$ to $t$.
\end{proof}

\section{Fixed-parameter tractability results}\label{sec-fpt}

\subsection{Preliminaries}
We first establish properties of shortest transforming sequences that allow for clean proofs of our results. 
C{\u{a}}linescu et al.~\cite{DBLP:journals/siamdm/CalinescuDP08} have shown that in any unlabelled undirected graph, it is possible to transform any configuration 
into another by a single move of each token; in Lemma~\ref{lemma-move-once}, we show that even for shortest transforming sequences, we can assume no token moves twice.

\begin{lemma}\label{lemma-move-once}
For any instance $(G, S, T, \ell)$ of \textsc{UUTM} or any instance $(D, S, T, \ell)$ of \textsc{UDTM}, there exists a shortest transforming sequence 
in which no token moves more than once.
\end{lemma}

\begin{proof}
  We provide the proof for the directed case, which implies the undirected case.
  We assume to the contrary that there exists an instance for which any shortest transforming sequence requires the move of at least one token a second time.  
  We minimize the length of the subsequence of moves $\alpha$ starting at the first move of a token and ending at the second move of the same token.
  We can represent $\alpha$ as $m_0, m_1, m_2, \ldots, m_{\ell-1}, m_\ell$, where the move $m_i = (s_i,t_i)$. 
  Since the first and last move are of the same token, $t_0 = s_\ell$.   

  To prove the lemma, it suffices to show that we can form a shorter sequence of moves, $\alpha'$, such that $\alpha$ and $\alpha'$ result in the same configuration and $\alpha'$ is shorter than $\alpha$.  We will refer to  moves of the form $(s_i,t_i)$ as \emph{matching moves} and moves of the form $(s_i,t_{i+1})$ as \emph{offset moves}; a \emph{sequence} of matching (offset) moves consists of moves for consecutive values of $i$.

  We construct $\alpha'$ as a sequence of \emph{segments} (Figure~\ref{fig-move-once}), where each segment is formed by adding zero or more valid offset moves until a blocked offset move $(s_x,t_{x+1})$ is reached, adding all valid matching moves starting at $(s_{x+1},t_{x+1})$ until a blocked matching move $(s_y,t_y)$ is reached, and then adding the \emph{linking move} $(s_x,t_y)$. The next segment will start by attempting to add the offset move $(s_y,t_{y+1})$.


  To prove that the moves in the segments are all valid, we first show that $(s_0,t_1)$, the first move in the first segment, is valid.
We know that in $\alpha$, $m_1$ is valid after $m_0$ is executed, but 
since $\alpha$ is of minimum length, $m_1$ cannot be executed before $m_0$.  Since $(s_0,t_0)$ results in a free path between 
$s_1$ and $t_1$, we know that $s_0$ must contain the only token on an intermediate vertex on a path from $s_1$ to $t_1$, and hence 
there are free paths from $s_0$ to each of $s_1$ and $t_1$.

To show that all of $\alpha'$ is valid, we prove the following statements by induction on the number of segments executed so far.
  \begin{enumerate}
  \item{} The first matching move in any segment is valid.
  \item{} The linking move in any segment is valid.
  \end{enumerate}

  To show that part (1) holds, we suppose instead that the matching move $(s_{x+1},t_{x+1})$ is blocked after the execution of $(s_{x-1},t_x)$. 
  By the definition of the algorithm, we are attempting to add $(s_{x+1},t_{x+1})$
because $(s_x,t_{x+1})$ is not valid.

Because the matching move $(s_{x+1},t_{x+1})$ is valid in $\alpha$, by Observation~\ref{obs-only-block}, we can then conclude that $s_x$ blocks the move, and that there is a free path from $s_x$ to $t_{x+1}$, contradicting the assumption $(s_x,t_{x+1})$ is not valid.

For part (2), we consider the placement of tokens in $\alpha$ and $\alpha'$ at the point at which $(s_y,t_y)$ is blocked.  By Observation~\ref{obs-only-block}, the fact that the $(s_y,t_y)$ is valid in $\alpha$ but not in $\alpha'$ implies that there is a free directed path from $s_x$ to $t_y$, and hence $(s_x,t_y)$ is valid.

  Finally, we complete the proof by showing that $\alpha$ and $\alpha'$ result in the same configuration.
  By the end of the first $\ell-1$ moves in $\alpha'$, the only difference between the configuration obtained by the end of the first $\ell-1$ moves in $\alpha$ is that there is no token on $t_0 = s_\ell$ and a token on $t_\ell$.  Thus, no further moves are required in $\alpha'$ to 
  achieve the configuration reached in $\alpha$ by the final move, $(s_\ell, t_\ell)$.  
\end{proof}

\begin{figure}[h]
\begin{centering}
\includegraphics[scale=.4]{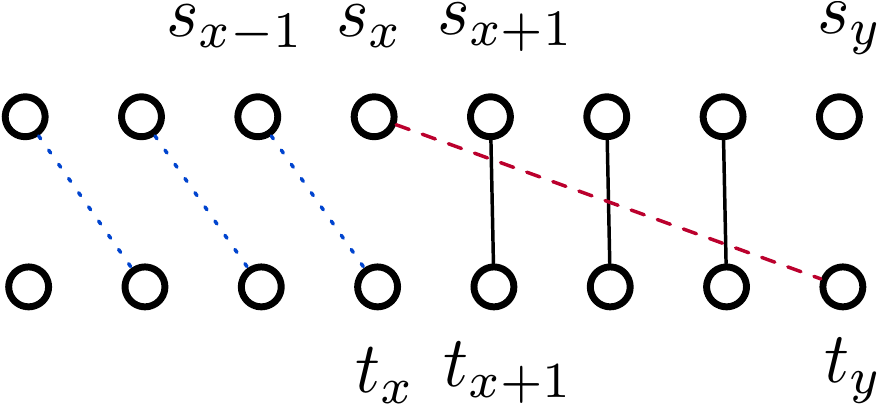}
\caption{Formation of a segment; sources and targets for consecutive values of $i$ are depicted by rows of circles, offset moves by dotted blue lines, matching moves by solid black lines, and the linking move by a dashed red line.}  
\label{fig-move-once}
\end{centering}
\end{figure}

We use Lemma~\ref{lemma-move-once} to show that we can find an equivalent instance in 
which every vertex is in $S \cup T$; we refer to such an instance as a \emph{contracted instance}. We use contracted instances to form algorithms when parameterizing by $k$.

\begin{lemma}\label{lemma-cliquify}
For any instance $(G, S, T, \ell)$ of \textsc{UUTM} or any instance $(D, S, T, \ell)$ of \textsc{UDTM}, we can form an equivalent 
contracted instance $(G', S', T', \ell)$ or $(D', S', T', \ell)$.
\end{lemma}

\begin{proof}
Since by Lemma~\ref{lemma-move-once} we know that a vertex $v \notin S \cup T$ is never the source vertex nor the target vertex 
of a move, the only role it can play is in connecting its neighbours.  We form $G'$ from $G$ by adding an edge between each pair of neighbours of each such vertex $v$.

To form $D'$, we instead add an arc from each in-neighbour of $v$ to each out-neighbour of $v$, where $v \notin S \cup T$.  Vertices corresponding to those in $S$ and $T$ form $S'$ and $T'$.
\end{proof}

\begin{corollary}\label{cor-kernel}
\textsc{UUTM} and \textsc{UDTM} admit an $O(k)$ vertex-kernel when parameterized by $k$. Moreover, the problems can be solved in $k^{O(\ell)} \cdot n^{O(1)}$ time. 
\end{corollary}

\begin{proof}
  The size of the kernel follows directly from Lemma~\ref{lemma-cliquify}.  
  To form the algorithm, we observe by Lemma~\ref{lemma-move-once} that it suffices to determine the set of obstacles $O$ that are sources of 
  moves, the pairing of sources and targets into moves, and the order of moves. 
  The number of moves $\ell$ will be at most $|S \setminus T| + |O|$, allowing us try all possible choices of vertices in $O$, pairings, and orders in the stated time bound.
\end{proof}  

\subsection{\textsc{Unlabelled Undirected Token Moving}}

Our algorithm for \textsc{Unlabelled Undirected Token Moving} relies on the characterization of the graph $G_{\alpha}$ of a transforming sequence $\alpha$ of minimum length of a contracted instance.  
In Lemmas~\ref{lemma-forest} and \ref{lemma-steiner}, we show that there exists $\alpha$ such that $G_{\alpha}$ is a forest of minimum Steiner trees.
By considering all possible ways of partitioning vertices in $S \Delta T$ into trees, and counting the number of moves required by 
each choice, in Theorem~\ref{thm-fpt-ell} we are able to obtain an FPT algorithm for the \textsc{UUTM} problem 
parameterized by $\ell$ on contracted instances, and hence by Lemma~\ref{lemma-cliquify}, for all instances. 


\begin{lemma}\label{lemma-forest}
For any contracted instance $(G, S, T)$ of \textsc{UUTM}, there exists a transforming sequence $\alpha$ of minimum length such that $G_{\alpha}$ is a forest.
\end{lemma}

\begin{proof}
  Suppose to the contrary that each transforming sequence of minimum length induces a graph containing a cycle.  
  Of these, we consider the minimum length of the subsequence $\beta$ of $\alpha$ consisting of the moves bracketing the 
  formation of the first cycle in $G_{\alpha}$: for the first cycle $C$ formed, we start with the first move such that the path for the move contains a 
  vertex in $C$, and end with the first move that completes $C$.  We use $\gamma$ to refer to the sequence 
  formed by removing the last move from $\beta$, and $(s,t)$ to refer to the last move in $\beta$.

  In the graph $G_{\beta}$, we use $S_0$ to denote the configuration before $\beta$ is executed, $S_{\gamma}$ to denote the 
  configuration after $\gamma$ is executed, and $S_{\beta}$ to denote the configuration after $\beta$ is executed (Figure~\ref{fig-forest} (a)). 
  By definition, $S_{\gamma}$ and $S_{\beta}$ differ only in the placement of a single token, which can be found on $s$ in $S_{\gamma}$ and on $t$ in $S_{\beta}$. 

  Since $(s,t)$ completes the cycle $C$, we can decompose $C$ into two paths between a pair of vertices $s'$ and $t'$, where $s'$ and $t'$ are both on the path for the move $(s,t)$.  We use $P$ to denote the path between $s'$ and $t'$ that is part of $G_{\gamma}$ and $P'$ to denote the path between $s'$ and $t'$ that is part of the path for the move $(s,t)$ (Figure~\ref{fig-forest} (b)).

\begin{figure}[h]
\begin{centering}
\includegraphics[scale=.3]{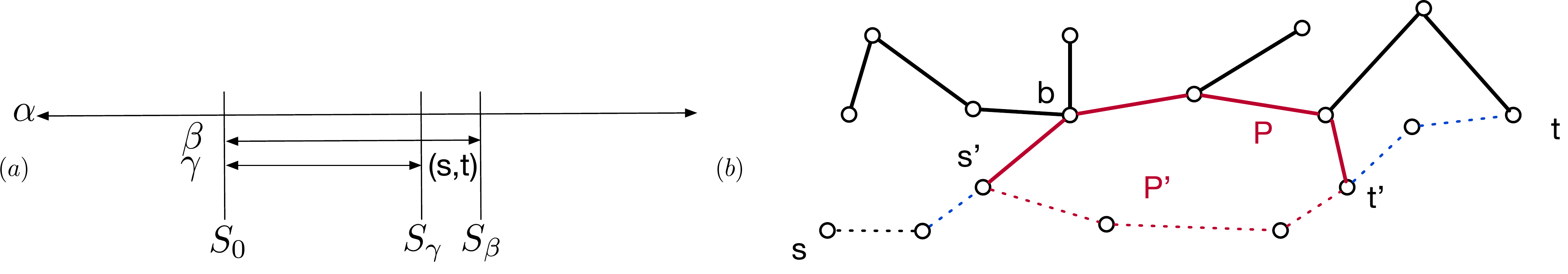}
\caption{(a) Depiction of sequences by horizontal lines and timing of configurations by vertical lines (b) $G_\beta$ with the dotted path for the move $(s,t)$;
$G'$ consists of all solid edges as well as the path from $s$ to $s'$.  No extra dotted edges are added to connect $t$, as it is already part of $G_{\gamma}$.
}
\label{fig-forest}
\end{centering}
\end{figure}

  Our goal is to form an alternative sequence of valid moves such that $P$ can be used instead of $P'$ to reach configuration $S_{\beta}$.  To make this possible, we identify the token on a vertex $b$ on $P$ in $S_{\gamma}$ that is closest to $s$ in $G_\beta$. We will show that we can find a sequence of moves $\gamma'$ that when executed on configuration $S_0$, results in a configuration $S_{\gamma}'$ that differs from $S_{\gamma}$ only in the placement of a single token, where $S_{\gamma}'$ has a token on $t$ and $S_{\gamma}$ has a token on $b$.  We can then execute the move $(s,b)$ to reach configuration $S_{\beta}$ from $S_{\gamma}'$.

  To obtain a contradiction and complete the proof, we need to show that $\gamma'$ uses no intermediate vertices of $P'$ and that the length of $\gamma'$ is no greater than the length of $\gamma$.  We form a tree $G'$ that we  initialize to be $G_{\gamma}$, and for each of $s$ and $t$ that is not in $G_{\gamma}$, we add to $G'$ a subpath of the path for the move $(s,t)$ with one endpoint in $G_{\gamma}$, the other endpoint the vertex $s$ or $t$, and all intermediate vertices (if any) not in $G_{\gamma}$.  The fact that $(s,t)$ closes the cycle $C$ guarantees that if $s$ and $t$ are not in $G_{\gamma}$, such paths exist, and that they contain no intermediate vertices of $P$ (Figure~\ref{fig-forest} (b)). 

To compare the lengths of $\gamma$ and $\gamma'$, we first observe that by construction, every token in $G_{\gamma}$ must be moved in order to transform $S_0$ into $S_\gamma$.  
Since $G'$ is a tree, we can invoke the polynomial-time algorithm of
C{\u{a}}linescu et al.~\cite{DBLP:journals/siamdm/CalinescuDP08}, allowing us to obtain a shortest sequence of valid moves from $S_0$ to $S_{\gamma}'$ in $G'$ to form $\gamma'$.  Since the modification of $G_{\gamma}$ to $G'$ did not increase the number of tokens, by Lemma~\ref{lemma-move-once}, the number of moves in $\gamma'$ cannot be greater than the number of tokens, and hence not greater than the number of moves in $\gamma$, as required to complete the proof.
\end{proof}  

\begin{lemma}\label{lemma-steiner}
  For a contracted instance $(G, S, T)$ of \textsc{UUTM} and a transforming sequence $\alpha$ of minimum length such that $G_{\alpha}$ is a forest, each tree in the forest is a minimum Steiner tree with terminals and leaves in $S \Delta T$ and internal vertices in $S \cup T$, and such that each internal vertex in $O$ is the source vertex of a move.  
\end{lemma}

\begin{proof}
The fact that the leaves of the tree are in $S \Delta T$ follows from the minimality of the length of $\alpha$.  Because each token in $G_{\alpha}$ must be moved in a contracted instance $(G, S, T)$, a minimum Steiner tree will contain the minimum number of vertices and hence the minimum number of moves.
\end{proof}

\begin{theorem}\label{thm-fpt-ell}
\textsc{UUTM} is in FPT when parameterized by $\ell$.
\end{theorem}

\begin{proof}
  We first form an equivalent contracted instance (Lemma~\ref{lemma-cliquify}), and then attempt all possible partitions of vertices in $S \Delta T$ into 
  at most $\ell$ Steiner trees, starting first with a single tree, then two, and so on.

  When considering use of $d$ trees, we first consider all possible ways of partitioning the vertices of $S \Delta T$ into $d$ groups, where each group 
  has equal number of vertices in $S \setminus T$ and $T \setminus S$, and then run the FPT Steiner tree algorithm~\cite{DreyfusWagner} on 
  each such set of vertices.  Because in each Steiner tree each token must move, the number of moves associated with each tree $\cal T$ will 
  be $|(S \setminus T) \cap V({\cal T})| + |O \cap V({\cal T})|$.  
  If the total number of moves is at most $\ell$ we have a yes-instance.  
  If we have not succeeded using $d$ trees to verify that the instance is a yes-instance, we try with the next value of $d$. 
  If none of the values succeed, we conclude that $(G, S, T, \ell)$ is a no-instance.

  The correctness of the algorithm follows from Lemmas~\ref{lemma-forest} and \ref{lemma-steiner}.
\end{proof}

\subsection{\textsc{Unlabelled Directed Token Moving}}

Like in the case of undirected graphs, our algorithm for the directed case relies on the characterization of 
the graph $D_{\alpha}$ of a transforming sequence $\alpha$ of minimum length of a contracted instance. 
We show, in Lemma~\ref{lemma-directed-forest}, that for any yes-instance there exists an $\alpha$ such that $D_{\alpha}$ is a 
directed forest.  As a replacement for the Steiner tree approach, the fact that we can bound the size of $D_{\alpha}$ suggests the use of the machinery of color coding, introduced by Alon et al.~\cite{DBLP:reference/algo/AlonYZ08} (similar to a result obtained by Plehn and Voigt~\cite{DBLP:conf/wg/PlehnV90}), to determine whether $D$ contains a labelled subgraph of the correct form to be $D_{\alpha}$ for a  contracted yes-instance.  Unfortunately, Theorem~\ref{thm-alon} cannot be used directly; in Theorem~\ref{thm-dir-ell} we adapt the technique for our purposes.

\begin{theorem}[\cite{DBLP:reference/algo/AlonYZ08}]\label{thm-alon}
Let $H$ be a directed forest on $q$ vertices. Let $D = (V,E)$ be a
directed $n$-vertex graph. A subgraph of $D$ isomorphic to $H$, if one exists, can be found 
in $2^{O(q)} \cdot n^{2} \cdot \log n$ worst-case time. Moreover, if a real-weight function $\beta: E \rightarrow \mathcal{R}$ is defined on the 
edges of $D$, then the algorithm can be adapted to find the copy of $H$ in $D$ with the maximal total weight. 
\end{theorem}

After removing extraneous vertices (Lemma~\ref{lem-clean-stuff}), 
we demonstrate that $D_\alpha$ forms a forest (Lemma~\ref{lemma-directed-forest}); to show that we can 
ignore cycles, we focus on a minimal graph containing a cycle as a counterexample.
More formally, we call a directed graph $D$ a \emph{circle graph} if $D$ is connected and the 
vertices in $V(D)$ can be partitioned into \emph{cycle vertices}, forming a simple cycle $C$ in the 
underlying undirected graph, and \emph{forest vertices}, forming a forest of trees attached to the cycle vertices, where in each tree either all arcs are directed towards the 
root or all arcs are directed away from the root.
An instance $(D, S, T, \ell)$ of \textsc{UDTM} is said to be a \emph{contracted circle instance} whenever $D$ is a circle graph and $S \cup T = V(D)$. 

\begin{lemma}\label{lem-clean-stuff}
  For $(D, S, T, \ell)$ a contracted instance of \textsc{UDTM} and $v \in S \cap T$, $(D, S, T, \ell)$ is a yes-instance if and only if 
  $(D - v, S \setminus \{v\}, T \setminus \{v\}, \ell)$ is a yes-instance when any of the following conditions hold:
  \begin{enumerate}
  \item{} There is no directed path from any vertex in $S \setminus T$ to $v$.
  \item{} There is no directed path from $v$ to any vertex in $T \setminus S$.
  \item{} Every directed path from any vertex in $S \setminus T$ to $v$ contains at least $\ell + 1$ obstacles and every directed path from $v$ to any
vertex in $T \setminus S$ contains at least $\ell + 1$ obstacles. 
  \end{enumerate}
\end{lemma}  

\begin{proof}
By Lemma~\ref{lemma-move-once}, in any shortest transforming sequence, either $v$ is the 
source of a clearing move and the target of a filling move, or it blocks the move of another token. 
Because each pair of clearing and filling moves forms a directed path from a vertex in $S \setminus T$ to a 
vertex in $T \setminus S$ (as does any move with source vertex in $S \setminus T$ and target 
vertex in $T \setminus S$), under the first two conditions $v$ can neither be part of such a move nor block such a move.

For the third condition, we observe that $\ell$ moves are not sufficient either to move enough other obstacles for there 
to be a clear path between $v$ and a vertex in $T \setminus S$, nor for $v$ to be the only obstacle blocking a path.
\end{proof}

\begin{lemma}\label{lem-circle-instance}
If there exist instances $(D, S, T)$ of \textsc{UDTM} such that for every transforming sequence $\alpha$ of minimum length, 
$D_{\alpha}$ is not a forest, then at least one of those instances must be a contracted circle instance. 
\end{lemma}

\begin{proof}
Of all those instances, let us consider an instance $(D, S, T)$ where $|V(D)|$ is minimized, $|E(D)|$ is minimized, 
and the length of a shortest transforming sequence $\alpha$ is minimized (in that order). We call such an instance a \emph{minimal instance}.

We first show that in a minimal instance,  the last move in the shortest transforming sequence is the one that induces one or more cycles in $D_{\alpha}$.
Suppose to the contrary that the first cycle is formed by a move $m$ that is not the last in the sequence. 
We then define $S_m$ to be the configuration after $m$ is executed. 
The instance $(D_{\alpha}, S, S_m)$ satisfies the statement of the lemma and the shortest transforming sequence consists 
of at least one fewer moves, contradicting the minimality of the instance.


We now show that $D_{\alpha}$ cannot contain more than one connected component nor more than one cycle. 
If instead $D_{\alpha}$ contains more than one connected component, since the moves in each component are independent of each other, the 
moves forming the components not containing the cycle could have all taken place after the last move in $\alpha$, contradicting 
the minimality of the length of $\alpha$. 
If the last move $m_\ell = (s_\ell,t_\ell)$ closes more than one cycle then we let $D'$ be the 
graph induced on all moves prior to $m_\ell$ and we let $P_m$ denote the path induced by $m_\ell$. 
We let $v_0$, $v_1$, and $v_2$ be the first, second, and third intersections of $P_m$ with $D'$ such that $v_0v_1 \not\in E(D')$ and
$v_1v_2 \not\in E(D')$, respectively. Such vertices must exist since $m_\ell$ is assumed to close more than one cycle. 
We note that we could have longer paths between $v_0$ and $v_1$ or between $v_1$ and $v_2$ but the same argument still applies. 
We create a new instance by deleting all edges between $v_0$ and $v_1$ and all edges between $v_1$ and $v_2$ and we instead add an edge from $v_0$ to $v_2$; 
since the path from $v_0$ to $v_2$ has at least three edges, we have found a 
smaller instance that uses fewer edges, contradicting the minimality of $|E(D)|$. 
Hence, in a minimal instance, $D_{\alpha}$ consists of a single component containing a single cycle.

We can assume that each vertex is the source or target of a move, as the transformation in the proof of Lemma~\ref{lemma-cliquify} results 
in a connected graph with a single cycle. 
Moreover, due to the minimality of $\alpha$, $\alpha$ cannot contain any move that does not pass through $C$. 
Consequently, any vertex outside of $C$ must either be connected by a directed path to a vertex in $C$ or connected 
by a directed path from a vertex in $C$, forming a contracted circle instance.
%
\end{proof}

We use Lemma~\ref{lem-circle-instance} in the proof of
Lemma~\ref{lemma-directed-forest}, where we use the structure of a
circle graph to form a transforming sequence.  We number the cycle
vertices as $v_1$, $v_2$, $\ldots$, $v_q$ in clockwise order, observe
that for each $i$, either $(v_i,v_{i+1})$ or $(v_{i+1},v_i)$ is an arc
in $D$. We say that a subsequence of vertices $v_a, v_{a+1}, \ldots,
v_b$ is a \emph{forward cycle segment} if $(v_i,v_{i+1})$ is an arc for each
$a \le i < b$ and a \emph{backward cycle segment} if $(v_{i+1},v_i)$ is an
arc for each $a \le i < b$.  Since the paths of moves are directed
paths, no move can use edges in more than one cycle segment.

Of particular interest are the \emph{junction vertices} shared by two
consecutive cycle segments, where each is either a source or a sink in
both cycle segments.  We refer to the \emph{in-pool} of a source
junction vertex $v$ as the set containing $v$  and any tree vertex that can reach $v$ by a directed path, and the \emph{out-pool} of a sink junction vertex $v$ as the set containing $v$ and any tree vertex that can be reached by a directed path from $v$.  Since each cycle segment starts and ends with a junction vertex, we can refer without ambiguity to the in-pool and out-pool of a cycle segment.  By partitioning the moves by cycle segments, we form a sequence of directed trees, thereby allowing us to apply Lemma~\ref{lem-dir-tree}.

\begin{lemma}\label{lem-dir-tree}
Given a directed tree $D$, two configurations $S$ and $T$ of $D$ such that every leaf of $D$ is in $S \Delta T$, 
and a one-to-one mapping $\mu$ from $S$ to $T$ such that there is directed path from each $s \in S$ to $\mu(s) \in T$ (and $s \neq \mu(s)$ for all $s$), then 
there exists a transformation from $S$ to $T$ in $D$.
\end{lemma}

\begin{proof}
We proceed by induction on the size of the tree $D$, observing that if 
$|V(D)| = 2$ then the statement clearly holds. 

If there exists a leaf $v \in T \setminus S$, we let $u = \mu^{-1}(v)$ and we let $P_{uv}$ be the unique directed path from $u$ to $v$. 
We let $w$ be the closest vertex to $v$ on $P_{uv}$ such that $w \in S$. Note that $w$ could be equal to $u$. Moreover, all 
vertices in $P_{wv}$ are in $T \setminus S$. Then, we form a smaller instance $D'$ and configurations $S'$ and $T'$ as follows. 
If $w \neq u$ then we delete $v$ from $D$ (and $T$), we remove $w$ from $S$, and we set $\mu(u) = \mu(w)$; since we have a directed path 
from $u$ to $w$ and a directed path from $w$ to $\mu(w)$, we have a directed path from $u$ to $\mu(w)$.  
If $w = u$ then we simply delete $v$ from $S$ and if $w$ becomes a leaf we also delete $w$. When $w$ does not become a leaf it becomes a free vertex (not in $S \cup T$). 
To form a transforming sequence from $S$ to $T$ in $D$, we start with the move $(w,v)$ and then add the transforming sequence for the smaller instance.

Now we consider the situation in which all the leaves of $D$ are in $S \setminus T$. For a leaf $u \in S \setminus T$, we let 
$v = \mu(u)$ and, again, we let $P_{uv}$ be the unique directed path from $u$ to $v$. 
We find the closest vertex $w$ to $v$ on $P_{uv}$ such that $w \in T$. All the vertices in $P_{uw}$ are in $S \setminus T$. 
We let $\mu(w) = x$ (assuming $w \in S \cap T$) and $\mu^{-1}(w) = y$. We can form a smaller instance by removing $u$, and form a 
transforming sequence from the transforming sequence of the smaller instance followed by the move $(u,w)$. 
To do so, we set $\mu(y) = x$ and $\mu(w) = v$; by construction, there is a directed path from $w$ to $v$, from $w$ to $x$, and from $y$ to $w$. 
If $w \in T \setminus S$ then we let $\mu^{-1}(w) = y$ and form the smaller instance in a similar manner but setting $\mu(y) = v$. 
\end{proof}

\begin{figure}[h]
\begin{centering}
\includegraphics[scale=.3]{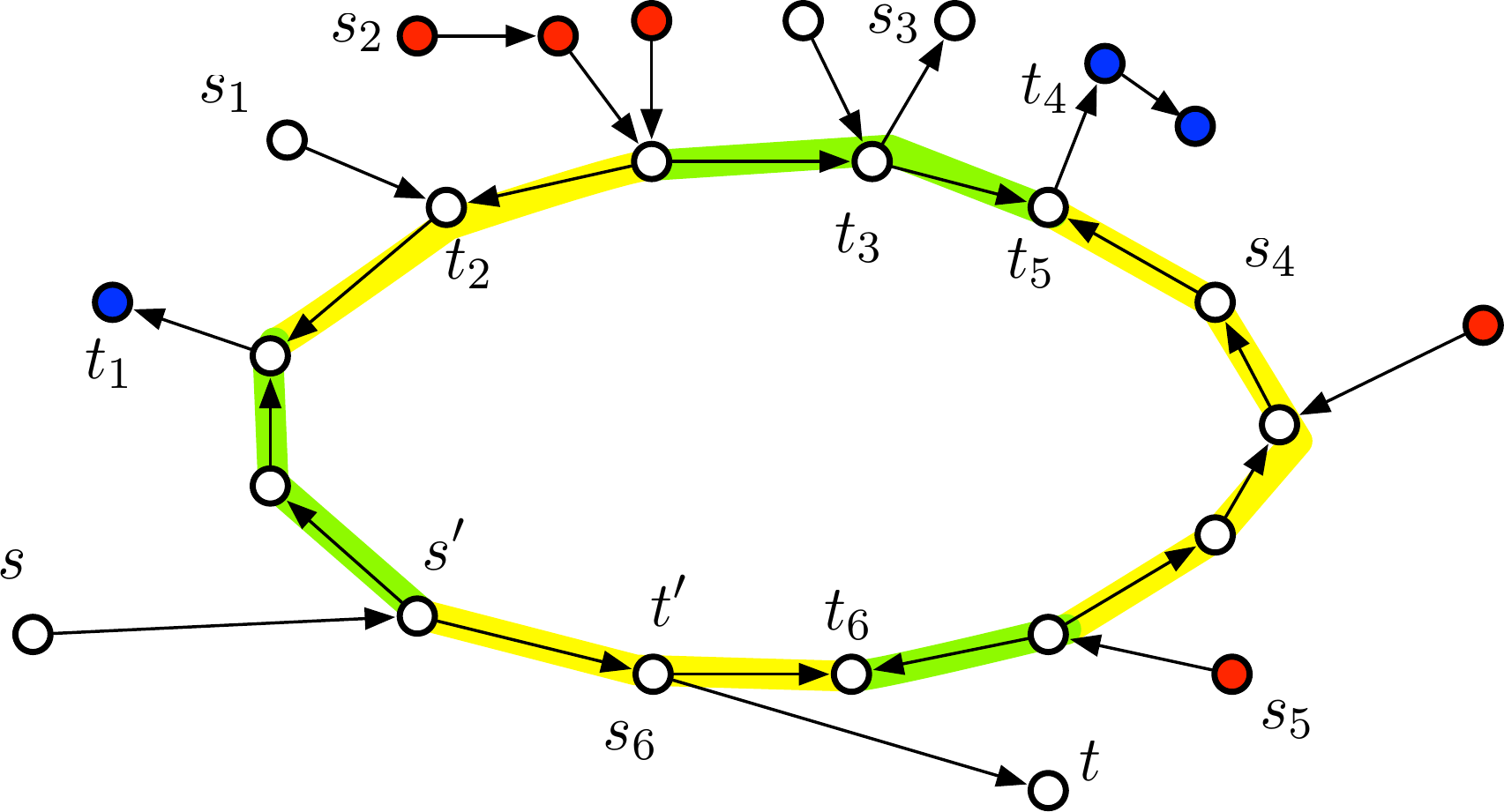}
\caption{Example circle graph with cycle segments highlighted, in-pools marked in red, and out-pools marked in blue, and sources and sinks of linking moves shown}
\label{fig-circlegraph}
\end{centering}
\end{figure}

\begin{lemma}\label{lemma-directed-forest}
For any contracted yes-instance $(D, S, T, \ell)$ of \textsc{UDTM}, there exists a 
transforming sequence $\alpha$ of minimum length such that $D_{\alpha}$ is a directed forest. 
\end{lemma}
\begin{proof}
For the sake of contradiction, we assume that $(D, S, T, \ell)$ is an instance for which for any $\alpha$ of minimum length, $D_{\alpha}$ contains a cycle. 
By Lemma~\ref{lem-circle-instance} we can assume that $(D, S, T, \ell)$ is a contracted circle instance, $(s,t)$ is the last move in $\alpha$, $D_{\alpha}$ contains a single cycle, $\beta$ is formed by removing $(s,t)$ from $\alpha$, and $D_{\beta}$ is a forest.


As in the proof of Lemma~\ref{lemma-forest}, we decompose the cycle into two undirected paths $P$ and $P'$ in the underlying undirected graph, both between the vertices $s'$ and $t'$, where $s'$ and $t'$ are both on the path for the move $(s,t)$.  We focus on the path $P$ that is not part of the move $(s,t)$.  We break $P$ into cycle segments, and show how to create a sequence of moves $\alpha'$ that results in the same configuration as $\alpha$ but does not use $P'$.  In forming $\alpha'$, we will find a sequence of linked moves, where for any two moves $m_i$ and $m_{i+1}$, we will say that $m_i$ and $m_{i+1}$ are \emph{linked} if there is directed path from $s_i$ to $t_{i+1}$.


In choosing linked moves, we consider only moves that use arcs of $P$.  Since each move can use arcs of only a single cycle segment, we can partition the moves into sets $\Sigma_1, \Sigma_2, \ldots, \Sigma_r$, where $\Sigma_1$ corresponds to the segment that includes $s'$ and $\Sigma_r$ the cycle segment that includes $t'$.  We will form a linked sequence starting with a move with source $s$, ending with a move with target $t$, and consisting of a sequence of moves for each cycle segment.

Because every arc in the cycle must be part of a move, for any forward cycle segment, we can find a sequence of moves such that the source of the first move is in the in-pool, the target of the last move is in the out-pool, and all moves are linked.   Similarly, for any backward cycle segment, we can find a sequence where the target of the first move is in the out-pool, the source of the last move is in the in-pool, and all moves are linked.  Further, we can link the first and last moves of the sequences of linked moves associated with consecutive cycle segments, 
since targets of moves in consecutive cycle segments using the same out-pool can be reached by the sources of both moves, and sources of moves in consecutive cycle segments using the same in-pool can reach the targets of both moves.

We now form a sequence $\alpha'$ of the same length as $\alpha$, but which does not use $P'$. By Lemma~\ref{lem-dir-tree}, it suffices to form a mapping $\mu$ between sources and sinks such that there is a directed path between each $v$ and $\mu(v)$.  We observe that there is a path from $s$ to a vertex in the out-pool of the first backward cycle segment, allowing us to map $s$ to the target of the first move for that cycle segment.  Using the definition of linked moves, for each pair of linked moves $m_i$ and $m_{i+1}$, we can
set $\mu(s_i)$ to $t_{i+1}$ for each $i$.
Finally, since there is path to $t$ from the in-pool of the last forward segment,
$\mu^{-1}(t)$ is set to the
source of the last move for that cycle segment.
\end{proof}

To check if $(D, S, T, \ell)$ is a contracted yes-instance, it suffices to determine whether or not a labelled version of $D$ contains a subgraph of the correct form to be $D_\alpha$.
We assign labels to vertices of $D$ such that the vertices of $S \setminus T$ are labelled from $s_1$ to $s_\Delta$, the vertices of $T \setminus S$ are labelled from $t_1$ to $t_\Delta$, and all other vertices are assigned label $\Delta$. Thus, we have $|S \Delta T| + 1$ distinct labels. We 
say that $D$ is \emph{($\Delta + 1$)-labelled} and use 
$\mbox{lab}(v)$ to denote the label of vertex $v$.
In $D_\alpha$, all vertices not in $S \Delta T$ receive label $\Delta$. 

When $(D, S, T, \ell)$ is a contracted yes-instance with $\alpha$ a transforming sequence of minimum length, $D_\alpha$ has at most $|S \Delta T| + \ell - |S \setminus T| = \ell + |S \setminus T| \leq 2\ell$ vertices, 
as otherwise $\ell + 1$ moves are required.  Thus, 
we can enumerate 
all possible $(\Delta + 1)$-labelled directed graphs of size at most $2\ell$ and check whether any one of them implies a yes-instance and can be found as a subgraph of $D$.
For $H$ a $(\Delta + 1)$-labelled directed forest,
we let $S' = \{v \in V(H) \mid \mbox{lab}(v) \in \{s_1, \ldots , s_\Delta, \Delta\}\}$ and $T' = \{v \in V(H) \mid \mbox{lab}(v) \in \{t_1, \ldots , t_\Delta, \Delta\}\}$. Then, $H$ is said to be a \emph{witness} for 
$(D, S, T, \ell)$ if $(H, S', T', \ell)$ is a yes-instance and a subgraph isomorphic to $H$ can be found in $D$ such that the labelling of the vertices 
in $S \Delta T$ is respected.

\begin{theorem}\label{thm-dir-ell}
\textsc{UDTM} is in FPT when parameterized by $\ell$.
\end{theorem}

\begin{proof}
  For any instance $(D, S, T, \ell)$ of \textsc{UDTM}, we can transform the instance using Lemmas~\ref{lemma-cliquify} and~\ref{lem-clean-stuff} to form a contracted instance such that for every vertex $v \in S \cup T$ (that is, every vertex $v$ in the graph), $v$ is on a directed path from a vertex in $S \setminus T$ to a vertex in $T \setminus S$.

By Lemma~\ref{lemma-directed-forest}, we know that if $(D, S, T, \ell)$ is a yes-instance, then there exists 
a transforming sequence $\alpha$ of minimum length such that $D_{\alpha}$ is a directed forest. 
For any contracted yes-instance of \textsc{UDTM}, at most $\ell - |S \setminus T|$ moves can be of tokens outside of $S \setminus T$.  Hence, the total number of vertices in $D_{\alpha}$ is at least $|S \Delta T|$ and at most $|S \Delta T| + \ell - |S \setminus T|$.

To check if $(D, S, T, \ell)$ is a yes-instance, we search for a witness
by enumerating all $(\Delta + 1)$-labelled directed forests on $q$ vertices, for each value $|S \Delta T| \leq q \leq |S \Delta T| + \ell - |S \setminus T|$.
We observe that the total number of such forests is bounded by a function of $\ell$.  For each forest $H$, we determine whether $H$ is a witness by determining whether $(H, S', T', \ell)$ is a yes-instance (which can be achieved in FPT by brute force) and whether there is a subgraph of $D$ isomorphic to $H$ that respects the labels.

To check subgraph isomorphism, we create edge-weighted digraphs $D'$ and $H'$ with vertices corresponding to those in $D$ and $H$. We assign all edges in $D'$ weight one.  Each vertex in $D'$ and $H'$ corresponding to a vertex with label $s_i$ is given $i$ degree-one in-neighbors, connected by arcs with weight $iq$, and each vertex in $D'$ and $H'$ corresponding to a vertex with label $t_j$ is given  $|S \setminus T| + j$ degree-one out-neighbors, connected by arcs with weight $(|S \setminus T| + j)q$.

We can now invoke Theorem~\ref{thm-alon} on $D'$ and $H'$. If $D$ is isomorphic to $H$, then the copy of $H'$ in $D'$ with maximal total weight will include all the edges of weight greater than one, thereby providing a label-preserving isomorphism; otherwise, we have a no-instance. 
\end{proof}

\section{Hardness results}\label{sec-hardness}

\subsection{Preliminaries}

To strengthen some of our hardness results, we prove that for any instance of \textsc{UDTM} we can find an equivalent instance that is of 
degree at most three (Lemma~\ref{lemma-degree}) or 2-degenerate (Lemma~\ref{lemma-subdivide}). 

\begin{lemma}\label{lemma-degree}
For any instance $(D, S, T)$ of \textsc{UDTM}, we can form an equivalent instance $(D', S', T')$ such that $D'$ has maximum degree three. 
\end{lemma}

\begin{proof}
We form $D'$ from $D$ by constructing a gadget for each vertex, where the gadget for a vertex $v$ with in-degree $d_i$ and 
out-degree $o_i$ will consist of a directed path of $d_i + o_i + 1$ vertices (Figure~\ref{fig-broomextras} (a)). The first $d_i$ vertices 
are the \emph{in-vertices} of $v$, the next vertex is the \emph{central vertex} of $v$, and the remaining $o_i$ vertices are the \emph{out-vertices} of $v$.

For any arc $(x,y)$ in $D$, we connect the gadgets for $x$ and $y$ by adding an arc from an out-vertex 
of $x$ to an in-vertex of $y$.  Because each vertex has one in-vertex for each in-neighbor in $D$ and 
one out-vertex for each out-neighbor in $D$, each of the in-vertices and out-vertices have degree 
exactly three, and the central vertex has degree exactly two.

Finally, to form $S'$ and $T'$, we choose the central vertices for each vertex in $S$ and $T$, respectively. 
The equivalence of the instances follows directly from the fact that any path $P$ in $D$ corresponds to a path 
$P'$ in $D'$, where $P'$ passes through a sequence of central vertices corresponding to the intermediate vertices in $P$.
\end{proof}

\begin{lemma}\label{lemma-subdivide}
For any instance $(G, S, T)$ of \textsc{UUTM} or any instance $(D, S, T)$ of \textsc{UDTM}, we can form an equivalent 
instance $(G', S', T')$ or $(D', S', T')$ such that $G'$ or $D'$ is 2-degenerate.
\end{lemma}

\begin{proof}
We simply subdivide each edge any constant amount of times (maintaining direction in the case of directed graphs). 
By Lemma~\ref{lemma-move-once} we know that the ``new'' vertices are never the source vertex nor the target vertex of a move.
\end{proof}  

\subsection{Parameter $\ell - |S \setminus T|$}

\begin{theorem}\label{thm-broom}
The problems \textsc{UUTM}, \textsc{UDTM}, \textsc{LUTM}, and \textsc{LDTM}
are W[2]-hard when parameterized by $\ell - |S \setminus T|$.
\end{theorem}

\begin{proof}
  We demonstrate parameterized reductions from \textsc{Red-Blue Dominating Set}, known to be W[2]-hard~\cite{DF97}, which determines for a bipartite graph $G_D = (V_B \cup V_R, E)$ of blue and red vertices and an integer $k_D$, whether $G_D$ contains a subset $D \subseteq V_B$ of size at most $k_D$ such that each vertex in $V_R$ is the neighbor of a vertex in $V_B$.

Each of the reductions makes use of a different modification
to the ``broom graph'' construction of
C{\u{a}}linescu et al.~\cite{DBLP:journals/siamdm/CalinescuDP08}.
For  \textsc{LDTM}, we form a graph $D$ such that $V(D) = B \cup R \cup H \cup W$, where $B$ and $R$ contain one vertex for each vertex in $V_B$ and $V_R$, respectively, $|H| = |B|$, and $W$ is a directed path of length $|R|$.  There is an arc from the last vertex in $W$ to all vertices in $B$, arcs from the vertex in $B$ corresponding to $v \in V_B$ to all the vertices in $R$ corresponding to the neighbors of $v$ in $G_D$, and a cycle of length two between each vertex in $B$ and a distinct vertex in $H$ (Figure~\ref{fig-broomextras} (b)).

We set $S = W \cup B$, $T = B \cup R$, and $\ell = |R| + 2k_D$.  Since $|R| = |S \setminus T|$, $2k_D = \ell - |S \setminus T|$.

\begin{figure}[h]
\begin{centering}
\includegraphics[scale=.3]{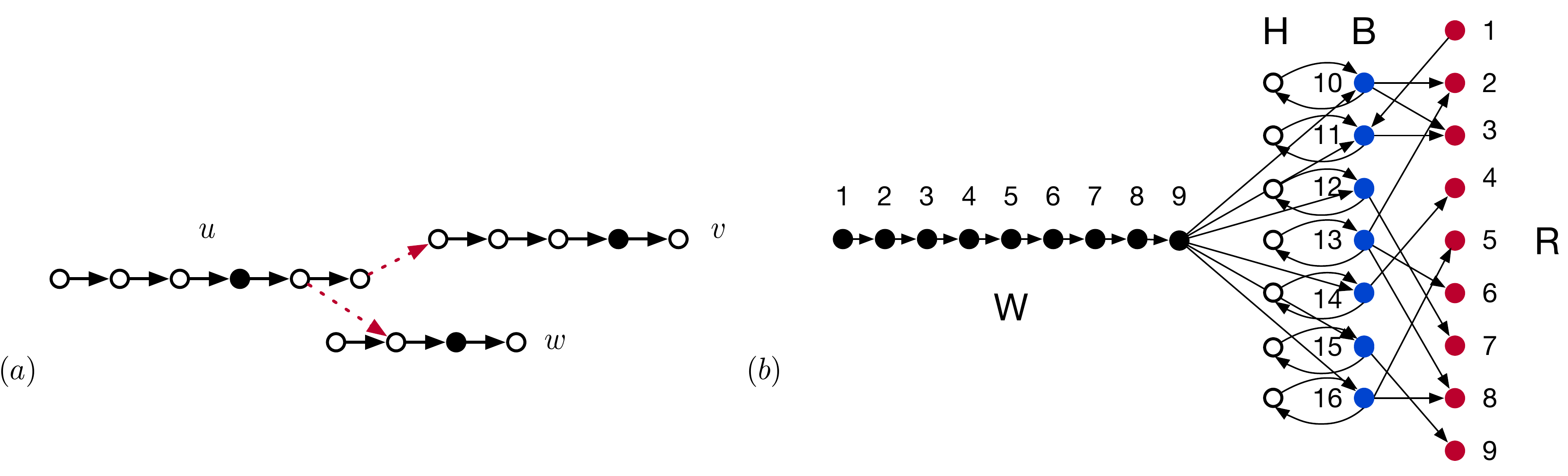}
\caption{(a) Gadgets for vertices $u$, $v$, and $w$ as well as the connections for the arcs $(u,v)$ and $(u,w)$
  (b) Broom graph reduction for \textsc{LDTM}, where $S = W \cup B$ and $T = B \cup R$}
\label{fig-broomextras}
\end{centering}
\end{figure}

By Lemma~\ref{lemma-move-once}, we can assume that the moves consist of $|S \setminus T|$ moves of tokens from $W$ to $R$, $k_D$ 
clearing moves, and $k_D$ filling moves; by construction, clearing moves will move tokens from $B$ to $H$.  A selection of $k_D$ 
vertices in $B$ for the clearing and filling moves will be possible only when $(G_D, k_D)$ is a yes-instance 
of \textsc{Red-Blue Dominating Set}, since otherwise it will not be possible to create free paths from all of $W$ to all of $R$.

The reductions for the remaining problems are similar.  For \textsc{LUTM}, all edges are undirected; it 
thus suffices to have a single edge between each vertex in $B$ and its corresponding vertex in $H$. 
For \textsc{UUTM} and \textsc{UDTM}, $H$ is empty and $\ell = |R| + k_D$.  The $k_D$ tokens in $B$ are moved to $R$, and the 
remaining moves are $\ell - k_D$ moves from $W$ to $R$ followed by $k_D$ filling moves from $W$ to $B$.  As for the other reductions, a 
transforming sequence is possible only when $(G_D, k_D)$ is a yes-instance.
\end{proof}  

Finally, we use Lemmas~\ref{lemma-degree} and~\ref{lemma-subdivide} to obtain the following results.

\begin{corollary}
\textsc{UUTM} is W[2]-hard when parameterized by $\ell - |S \setminus T|$, even when restricted to 2-degenerate graphs.
\end{corollary}

\begin{corollary}
\textsc{UDTM} is W[2]-hard when parameterized by $\ell - |S \setminus T|$, even when restricted to 2-degenerate graphs of maximum degree three.
\end{corollary}

\subsection{Parameter $\ell + f$}

We now show that \textsc{LUTM} and \textsc{LDTM} are W[1]-hard parameterized by $\ell$ or $\ell + f$ on general graphs. 
We give reductions from the \textsc{Multicolored Subgraph Isomorphism} problem, which determines whether there is a subgraph 
of a vertex-colored graph $G_M$ that is isomorphic to a vertex-colored graph $H$. 
The problem is W[1]-hard when parameterized by solution size, even when $H$ is a 3-regular connected bipartite graph~\cite{DBLP:journals/toc/Marx10}. 
We define $H$ to be a connected 3-regular bipartite graph such that 
$V(H) = [c]$, and use $\mbox{col}(u) \in [c]$ to denote the color of a vertex $u \in V(G_M)$.

Both of our reductions create a \emph{node-vertex} $v(w)$ for each
$w \in V(G_M)$ and then use the structure of $H$ to group node-vertices by color to form \emph{supernodes}.  By judicious assignment of labels to tokens moving between supernodes, we ensure that reconfiguration can occur only if we can clear tokens on a set of node-vertices that form a subgraph of $G_M$ isomorphic to $H$.  We start with the directed case and then explain  modifications for the undirected case. \\

\subsubsection{\textsc{Labelled Directed Token Moving}}

We explain the reduction in three steps: (1) forming a DAG $H'$ from $H$ to provide the structure of $D$, (2) creating and connecting supernodes, and (3) adding gadgets to constrain movement of tokens.


\noindent \textbf{Step 1.} To form a DAG $H'$, we create a vertex $h_i \in V(H')$ for each vertex $i \in V(H)$, making use of breadth-first search to assign each vertex to a level.  Choosing an arbitrary vertex $r$ of $H$ for the sole vertex $h_r$ at level 0 in $H'$, we assign each remaining vertex $h_i$ to the minimum level $p$ such that there is a path of $p$ vertices from $r$ to $i$ in $H$. By adding the edges forming the breadth-first search tree and directing them from vertices at smaller levels to larger levels, we form a directed acyclic graph (Figure~\ref{fig-handhprime} (c)).   By our construction, each arc connects vertices in adjacent levels.

\begin{figure}[h]
\begin{centering}
\includegraphics[scale=.4]{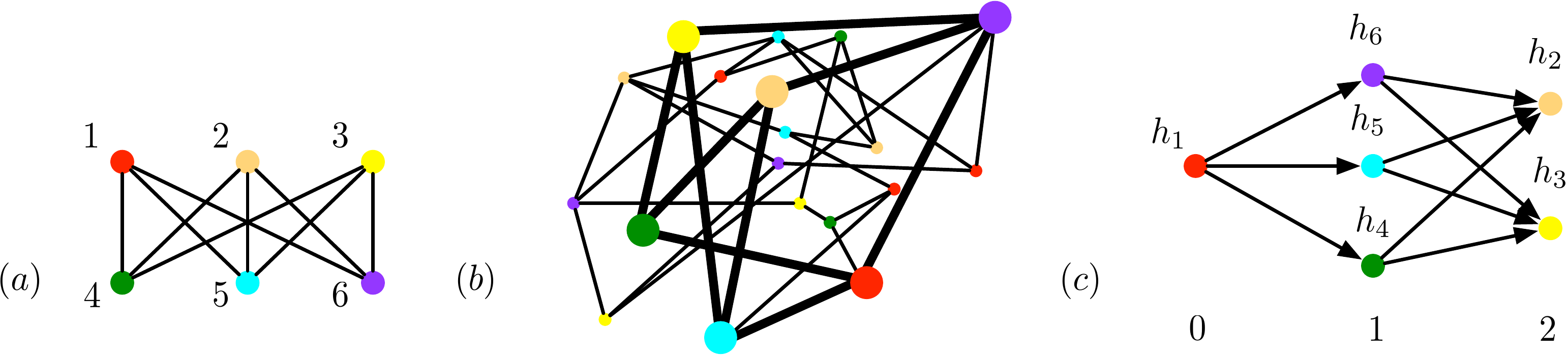}
\caption{(a) $H$ (b)  $G_M$, with a subgraph isomorphic to $H$ shown with larger circles and thicker lines (c) $H'$, where $r = 1$, showing level numbers.}
\label{fig-handhprime}
\end{centering}
\end{figure}

\noindent \textbf{Step 2.}  We can safely ignore any edge in $G_M$ between vertices of the same color or between colors not connected by an edge in $H$, as no such edge can form part of a subgraph of $G_M$ isomorphic to $H$.   For each color $i$, supernode $D_i$ consists of all node-vertices $v(w)$ such that $\mbox{col}(w) = i$ in $G_M$; we consider $D_i$ to be at the same level as $h_i$ in $H'$.  To form $D$, we add an arc between any vertex $v(x) \in D_i$ and $v(y) \in D_j$ such that $(h_i,h_j)$ is an arc in $H'$ and $(x,y)$ is an edge in $G_M$ (Figure~\ref{fig-directed} (a)).  When there exists an arc between vertices in $D_i$ and $D_j$, we say that $D_i$ is a \emph{super-in-neighbor} of $D_j$ and that $D_j$ is a \emph{super-out-neighbor} of $D_i$.

\begin{figure}[h]
\begin{centering}
\includegraphics[scale=.32]{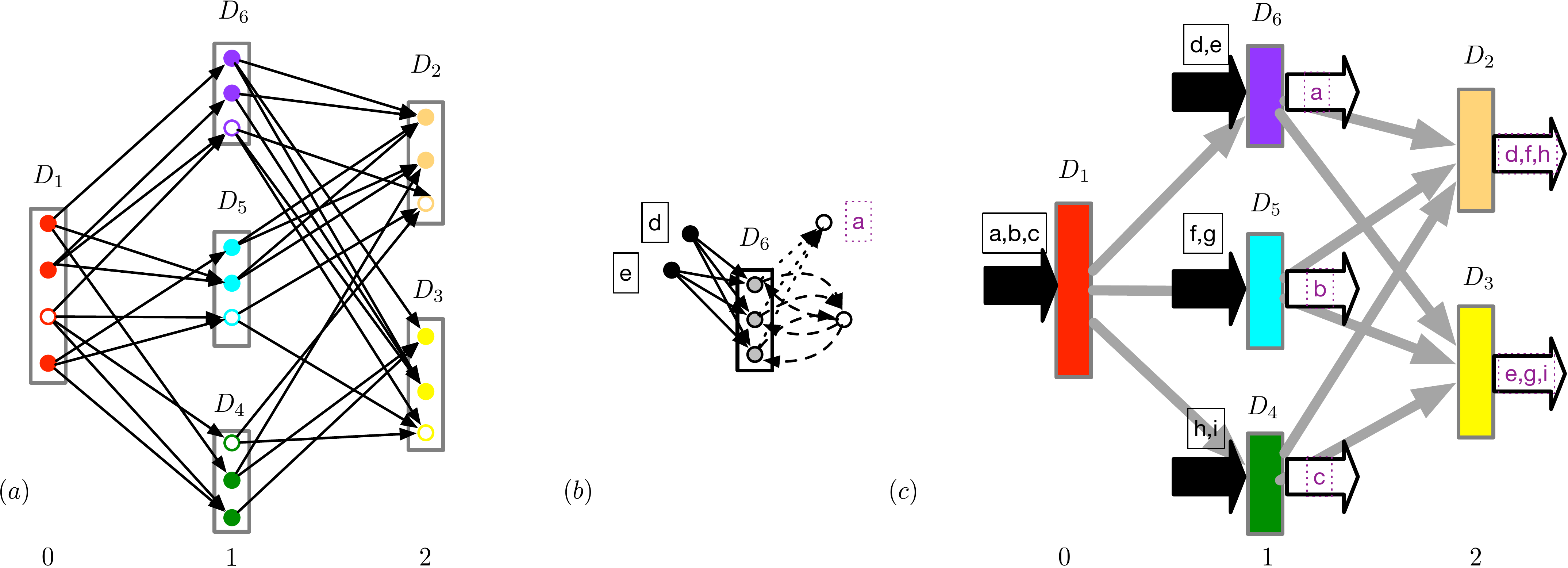}
\caption{(a) Supernodes of $D$ with node-vertices corresponding to a subgraph isomorphic to $H'$ shown as unfilled circles (b) Source, target, and storage gadgets for $D_6$ with edges shown using solid, dashed, and dotted arcs, respectively (c) A schematic diagram showing supernodes as rectangles and source gadgets and target gadgets as labelled black and white arrows with token labels. For example, he source gadget attached to $D_6$ has two tokens with target vertices in out-neighbors $D_2$ and $D_3$, and the target gadget attached to $D_6$ has one vertex, the target vertex of a token in in-neighbor $D_1$.
}
\label{fig-directed}
\end{centering}
\end{figure}

\noindent \textbf{Step 3.}  In order to ensure that each edge is traversed, we associate a token with each arc in $H'$, where for arc $(h_i, h_j)$, in $S$ the token is assigned to a \emph{source gadget} in $D_i$, and in $T$, the token is assigned to a \emph{target gadget} in $D_j$.  Accordingly, we choose token labels so that for each supernode, there is one token in its source gadget for each super-out-neighbor, and there is one vertex in its target gadget for each super-in-neighbor.  For each token in $S \setminus T$, its source vertex and target vertex are in consecutive levels. 

For \textsc{LDTM}, the source gadget attached to a supernode consists of a set of vertices with tokens connected by arcs into each node-vertex in the supernode and the target gadget consists of a set of vertices without tokens connected by arcs from each node-vertex in the supernode (Figure~\ref{fig-directed} (b)).
In addition, associated with each supernode is a \emph{storage gadget} consisting of a single vertex that is the in-neighbor and out-neighbor of every node-vertex in the supernode.  As their names suggest, the union of the vertices in the source gadgets equals $S \setminus T$ and the union of the vertices in the target gadgets equals $T \setminus S$.  

To complete the construction for \textsc{LDTM}, we set $S \setminus T $ to all vertices in source gadgets, $T \setminus S$ to all vertices in 
target gadgets, and $S \cap T$ to all node-vertices. By construction, $|S \setminus T| = |E(H)|$ and $f = |V(H)|$; we set $\ell = |E(H)| + 2k$.

\begin{lemma}\label{lem-forward-general}
If $(G_M,H)$ is a yes-instance of \textsc{Multicolored Subgraph Isomorphism}, then $(D,S,T,\ell)$ is a yes-instance of \textsc{LDTM}.
\end{lemma}

\begin{proof}
  We will specify a sequence of $\ell$ moves, consisting of $k$ clearing moves, $S \setminus T$ moves of vertices in $S \setminus T$ to vertices in $T \setminus S$, and finally $k$ filling moves.  
  
  Using the fact that $(G_H, H)$ is a yes-instance, we define $V_H$ as the set vertices of $G_M$ in the subgraph isomorphic to $H$. For our clearing moves, we move the tokens on $v(w)$ for each $w \in V_H$ to adjacent storage gadgets.  Since each vertex $w$ has a different color, each $v(w)$ is in a different supernode, and hence the storage gadgets are all distinct.

  Moreover, for each arc $(h_i,h_j)$ in $H'$, we have cleared the endpoints of an edge from $D_i$ to $D_j$ in $D$, allowing all the moves from $S \setminus T$ to $T \setminus S$ to take place.  Finally, the $k$ tokens moved in the clearing moves can be returned to their original locations.
\end{proof}

\begin{lemma}\label{lem-backward-general}
If $(D,S,T,\ell)$ is a yes-instance of \textsc{LDTM}, then $(G_M,H)$ is a yes-instance of \textsc{Multicolored Subgraph Isomorphism}.
\end{lemma}

\begin{proof}
Due to sizes of $S \setminus T$ and $\ell$, the transforming sequence consists of  $k$ clearing moves, 
$\ell - 2k$ moves of vertices in $S \setminus T$ to vertices in $T \setminus S$, and $k$ filling moves, although the order among them is unconstrained.  
As a consequence of the construction, we observe that 
before the move of a token from a source gadget attached to $D_i$ to a target gadget attached to $D_j$, there must be two clearing moves, one 
of a node-vertex $v(x)$ in $D_i$ and one of a node-vertex $v(y)$ in $D_j$ such that $(x,y) \in E(G_M)$.

To be able to clear enough vertices to allow all $S \setminus T$ moves of tokens from source gadgets to target gadgets, we must 
be able to select $k$ node-vertices to move such that for every edge $(h_i, h_j)$ in $H'$, there 
is an edge $(x,y) \in E(G_M)$ such that $\mbox{col}(x) = i$, $\mbox{col}(y) = j$, and the vertices $v(x)$ and $v(y)$ have both been cleared. 
By definition, the selected node-vertices correspond to the vertices in a subgraph of $G_M$ isomorphic to $H'$, and hence to $H$, proving 
that $(G_M, H)$ is a yes-instance of the \textsc{Multicolored Subgraph Isomorphism} problem.
\end{proof}

Combining Lemmas~\ref{lem-forward-general} and~\ref{lem-backward-general} with the fact that $\ell + f = O(|V(H)| + |E(H)|) = O(k)$, we obtain the following theorem. 

\begin{theorem}\label{thm-dir-free}
\textsc{LDTM} is W[1]-hard when parameterized by $\ell + f$. 
\end{theorem}

\subsubsection{\textsc{Labelled Undirected Token Moving}}

We use the same basic structure as in the previous reduction, but need extra machinery to ensure that a token moving from $D_i$ to $D_j$ is unable to find a route that avoids all edges corresponding to edges in $G_M$ with endpoints of colours $i$ and $j$.   To this end, we introduce superedges and a clock gadget, defined below, and for each arc $(h_i,h_j) \in E(H')$, specify the numbers of tokens to move from $D_i$ to $D_j$ and $D_j$ to $D_i$ based on the level.  Source gadgets, target gadgets, and storage gadgets all consist of vertices connected by a single undirected edge to all node-vertices in a supernode.

To construct $G$ for \textsc{LUTM}, we construct supernodes in the same way as in $D$, and although $G$ is undirected, use the 
terms level, super-in-neighbor, and super-out-neighbor based on the structure of $H'$.  
In addition, we form an \emph{edge-path} $p(x,y)$ of length $K$ (to be defined later) for each edge $(x, y)$ in $G_M$, and 
for each arc $(h_i,h_j)$ in $H'$, we form a \emph{superedge} $D_{i,j}$ consisting of all edge-paths $p(x,y)$ such 
that $\mbox{col}(x) = i$ and $\mbox{col}(y) = j$.  The level of $D_{i,j}$ is considered to be the same as the level of $D_i$. 
To connect the node-vertices and edge-paths in $G$ (Figure~\ref{fig-undirected}), 
for each superedge $D_{i,j}$ and each edge-path $p(x,y)$ in $D_{i,j}$, we add an edge from $v(x)$ to one end of $p(x,y)$ 
and an edge from the other end of $p(x,y)$ to $v(y)$.  

\begin{figure}[h]
\begin{centering}
\includegraphics[scale=.25]{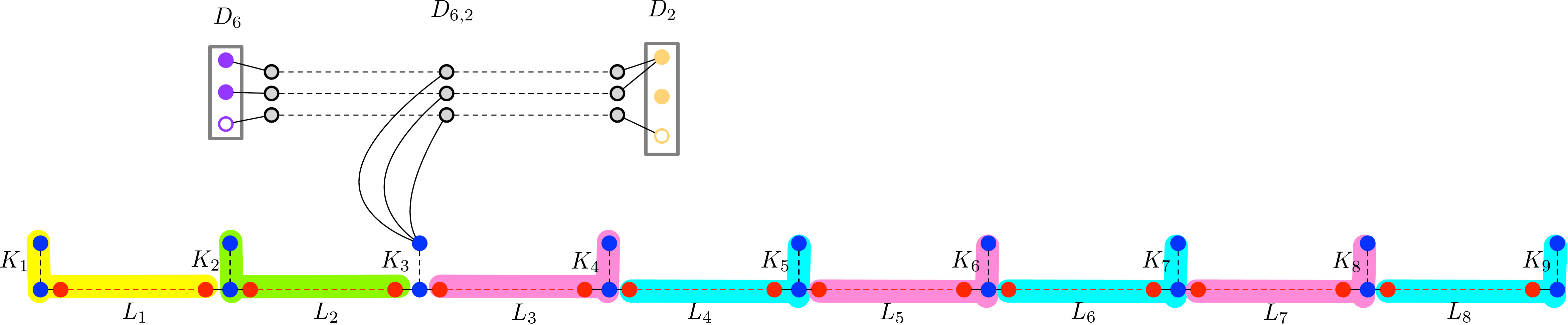}
\caption{Edge-paths for superedge $D_{6,2}$ connected to $K_3$ in the clock gadget (not to scale); the first two clock segments have moved to target vertices
}

\label{fig-undirected}
\end{centering}
\end{figure}

For convenience, we refer to the tokens in the source gadget attached to 
$D_i$ destined for the target gadget attached to $D_j$ as the \emph{$D_i$-$D_j$ tokens}.  Our proofs hinge on showing that the \emph{$D_i$-$D_j$} moves pass through superedge $D_{i,j}$.  To limit possible paths, for each supernode $D_{i,j}$ at level $r$, we create $\ell_r$ $D_i$-$D_j$ tokens and $\ell_r$ $D_j$-$D_i$ tokens, where for the last level $z$, $\ell_z = Q = 3k^2/2$ and for any level $y$, $\ell_y = Q \cdot \ell_{y+1}$.  

The \emph{clock gadget} is designed to allow the freeing of edge-paths in superedges one at a time in increasing order of level, which we 
call the \emph{clock numbering} (or just \emph{numbering}) of the superedges.  For large values $K$ and $L$ to be defined later, the clock 
gadget consists of $|E(H)|$ \emph{storage paths} $K_1, \ldots, K_{|E(H)|}$, each of length $K$,  and $|E(H)|-1$ \emph{linking paths} $L_1, \ldots, L_{|E(H)|-1}$, each of length $L$.  Referring to the two endpoints of each storage path as the \emph{top end} and the \emph{bottom end} and the two endpoints of each linking path as the \emph{left end} and the \emph{right end}, for each $i$, we add edges connecting the bottom end of $K_i$ to the right end of $L_{i-1}$ (if it exists) and to the left end of $L_{i}$ (if it exists).

We view tokens for the clock gadget in $S$ as being 
grouped into $|E(H)|-1$ \emph{clock segments}, each containing $K+L$ tokens, such that $K_1$ is free before the movement of the
 tokens in the first clock segment from source to target 
vertices, $K_2$ is free between the movements of the first and second clock segments, and $K_{|E(H)|}$ is free only after the movement of all of the clock segments. 
All vertices in $K_1$ are in $T \setminus S$, all vertices in $K_{|E(H)|}$ are in $S \setminus T$, and for clock segment $i$, the tokens are on $L_i$ and $K_{i+1}$ in $S$, in order from left to right and bottom to top, and on $K_i$ and $L_i$ in $T$, in order from top to bottom and left to right. We say that the \emph{clock is at position $p$} whenever the top end of $K_p$ does not have a token.  Finally, we connect the top end of each $K_i$ to the middle vertex in each edge-path in the superedge numbered $i$.




To complete the construction, we set $S$ to consist of vertices in the source gadgets, source vertices in the clock gadget, and all node-vertices and 
edge-paths, and $T$ to consist of vertices the target gadgets, target vertices in the clock gadget, and all node-vertices and edge-paths. The number of free vertices in $G$ is thus the total size of the storage gadgets, or $f = k$.

We now choose large enough values of $K$ and $L$ to control token movement.
We set $K = 2Q^* + k + 1$, where $Q^*$ is the total number of vertices in source gadgets in all supernodes.
To determine the number of moves $\ell$, we sum $(|E(H)| - 1)(K + L)$ moves for the $|E(H)|-1$ clock segments, $Q^*$ moves of tokens from source to target gadgets, $2k$ moves to clear and fill the $k$ node-vertices, and $2K|E(H)|$ to clear and fill one edge-path in each superedge, or $\ell = (|E(H)| - 1)(K + L) + Q^* + 2k + 2K|E(H)|$. 
We let $L = (|E(H)| - 1)K + Q^* +  2k + 2K|E(H)|  + 1$, so that 
written in a more convenient form, we have $\ell = |E(H)|L - 1$.

\begin{lemma}\label{lem-forward-undir}
If $(G_M,H)$ is a yes-instance of \textsc{Multicolored Subgraph Isomorphism}, then $(G,S,T,\ell)$ is a yes-instance of \textsc{LUTM}.
\end{lemma}

\begin{proof}
We specify a sequence of moves that clears one node-vertex in each supernode, completes all 
the moves between $S \setminus T$ and $T \setminus S$, and then fills the cleared node-vertices. 
As in the proof of Lemma~\ref{lem-forward-general}, we define $V_H$ as the set of vertices of $G_M$ 
in the subgraph isomorphic to $H$, and to clear node-vertices, we move the tokens on $v(w)$ for each $w \in V_H$ to adjacent storage gadgets.  

  Clearing node-vertices is not sufficient; the $D_i$-$D_j$ moves are still blocked by tokens in $D_{i,j}$. 
  For each superedge numbered $p$, we clear the edge-path corresponding to the edge in $V_H$ between 
  vertices of colors $i$ and $j$ by moving the tokens into $K_p$, execute the  $D_i$-$D_j$ moves, return the tokens 
  from $K_p$ to the edge-path, and then move the tokens in clock segment $p$ to their target vertices (except for when $p = |E(H)|$).  

Using the node-vertices corresponding to the subgraph of $G_M$ isomorphic to $H$, we can accomplish the transformation by $k$ clearing moves of the 
node-vertices, $2K$ moves for each superedge, $(|E(H)| - 1)(K + L)$ moves for the clock gadget tokens, $Q^*$ moves between source and target gadgets of supernodes,
and finally $k$ moves to fill the node-vertices.  
The total number of moves will be $\ell = (|E(H)| - 1)(K + L) + Q^* + 2k + 2K|E(H)|$.
\end{proof}

Before we prove the reverse direction, we first show that for a yes-instance of \textsc{LUTM}, the clock gadget behaves as required:  the clock can be in only one position at a time,  ``time'' cannot go backwards, and at any point in the transformation, we can have at most one superedge with an edge-path free of tokens. 

\begin{lemma}\label{lem-good-storage}
If $(G,S,T,\ell)$ is a yes-instance of \textsc{LUTM} then we cannot have $K_p$ and $K_{p'}$ such that the top ends of both are free of tokens. 
\end{lemma}
\begin{proof}
If there exist $K_p$ and $K_{p'}$ such that the top ends of both are free of tokens, then this implies that the tokens of at least 
one clock segment must move twice. Hence the number of moves required in the clock gadget will be at least $|E(H)|(K + L) > |E(H)|L - 1 = \ell$.
\end{proof}

\begin{lemma}\label{lem-good-clock}
If $(G,S,T,\ell)$ is a yes-instance of \textsc{LUTM}, then after the clock reaches position $p$ it can never go back to position $p - 1$ (or any earlier position). 
\end{lemma}
\begin{proof}
Assume that the clock reaches position $p$ and goes back to position $p - 1$. 
This implies that the clock would have to alternate from position $p$, to $p - 1$, and back to $p$. 
Hence, the total number of moves to transform the tokens on the clock to their target positions would be $(|E(H)| + 1)(K + L) > |E(H)|L - 1 = \ell$, a contradiction. 
\end{proof}

\begin{lemma}\label{lem-one-edge}
If $(G,S,T,\ell)$ is a yes-instance of \textsc{LUTM} then at any point in the transformation we can have at most one superedge with an edge-path free of tokens. 
Morever, whenever superedge numbered $p$ has a free edge-path, then the clock must be at position $p$. 
\end{lemma}

\begin{proof}
  By construction, each edge-path in a superedge contains $K = 2Q^* + k + 1$ tokens. Since the total number of vertices in source gadgets, target gadgets, and storage gadgets is $2Q^* + k$, at least one token of the free edge-path must have moved to the clock gadget. Since the top of $K_p$ is the only vertex in the clock gadget connected to the edge-path, the top of $K_p$ must have been free at the time of the move.
\end{proof}  

\begin{lemma}\label{lem-backward-undir}
If $(G,S,T,\ell)$ is a yes-instance of \textsc{LUTM}, then $(G_M,H)$ is a yes-instance of \textsc{Multicolored Subgraph Isomorphism}.
\end{lemma}

\begin{proof}
In counting the total number of moves in any transforming sequence, we need to allot at least $(|E(H)| - 1)(K + L)$ moves for transforming 
the clock gadget and at least $Q^*$ moves of vertices from source gadgets to target gadgets. 
Moreover, since each node-vertex blocks the moves from source gadgets to target gadgets, at least $2k$ moves are required to clear and then fill the node-vertices. 
Of the $\ell$ moves in total, at most $2K|E(H)|$ moves remain. 
If we can show that all of these moves are used to clear and fill an edge-path in each of the superedges, then we can conclude 
that there must be a subgraph of $G_M$ isomorphic to $H$. Hence, to complete the proof it suffices to show that each superedge 
is used for a move, namely, that for all $i$ and $j$, at least one $D_i$-$D_j$ or $D_j$-$D_i$ token passes through superedge $D_{i,j}$.

Suppose to the contrary that there exists a superedge $D_{i,j}$ that is not used for any move, and consider the first such superedge 
in clock numbering, say for $D_i$ at level $x$ and $D_j$ at level $x+1$. 
To determine the possible sequence of superedges in the path of $D_i$-$D_j$ tokens from $D_i$ to $D_j$ and $D_j$ to $D_i$, we observe 
that Lemmas~\ref{lem-good-storage},~\ref{lem-good-clock}, and~\ref{lem-one-edge} imply the following: the clock starts at position 1 and advances 
one position at a time;  
at any point in time, we can have at most one superedge with an edge-path free of tokens; and whenever superedge numbered $p$ has a free edge-path, 
then the clock must be at position $p$. Due to the clock numbering of superedges, we can view each superedge as \emph{opening} and then \emph{closing}; 
before opening and after closing, a superedge cannot be used for movement of tokens.  Furthermore, since the superedges open and close in order by level, 
we can refer to the \emph{level $x$ opening period} as the sequence of clock moves during which each superedge of level $x$ is opened and closed. 

Since tokens can move only between supernodes at levels $x$ and $x+1$ during the level $x$ opening period, 
the $D_i$-$D_j$ token must move to level $x + 1$ by the end of the level $x$ opening period, and then using at least one $x+2$ supernode 
(and possibly an alternating path of level $x+1$ and $x+2$ supernodes) during the level $x+1$ opening period to reach $D_j$. 
As for the $D_j$-$D_i$ token, it must move to its target position, i.e., $D_i$, using an alternating path of level $x$ and level $x+1$ supernodes. 
This also needs to happen before the end of the level $x$ opening period. 

We first show that $D_i$ must have exactly two super-out-neighbors at level $x+1$. 
Since no edge-path in a superedge at level $x-1$ or earlier can be opened in the future, the first edge in the path used by the $D_i$-$D_j$ tokens is 
from $D_i$ to a neighbor $D_{b} \ne D_j$ at level $x+1$; as a consequence, $D_i$ must have at least two neighbors at level $x+1$. 
The supernode $D_i$ cannot be the supernode at level 0, since there does not exist an alternating path of superedges at level 0 to be used to move the $D_j$-$D_i$ tokens; 
this implies that $D_i$ cannot have three super-out-neighbors. 

During the level $x$ opening period, the $D_j$-$D_i$ tokens must move through a sequence of level $x$ superedges, alternating between level $x$ and $x+1$ supernodes, 
until reaching $D_i$ from super-out-neighbor $D_b \ne D_j$ at level $x+1$. In order for the tokens to reach $D_b$ from $D_j$, they must pass through a neighbor 
$D_a \ne D_i$ of $D_b$ at level $x$ (possibly after a longer sequence of superedges). 
This implies that the superedge $D_{a,b}$ opens before the superedge $D_{i,b}$, allowing the tokens to move from $D_a$ to $D_b$ to $D_i$.

We now consider two cases, depending on whether $D_b$ has two or three super-in-neighbors at level $x$.

\noindent{\bf Case 1: $D_b$ has two super-in-neighbors} 
First, we show that it is impossible for both the $D_i$-$D_j$ tokens to reach $D_j$ and the $D_j$-$D_i$ 
tokens to reach $D_i$ during the level $x$ opening period.
Assume $D_j$ has two super-in-neighbors at level $x$. Let $D_c$ be the other neighbors of $D_j$ at level $x$.
Then, for the $D_i$-$D_j$ tokens to move from $D_i$ to $D_j$ during the level $x$ opening 
period, the superedge $D_{i,b}$ must open before the superedge $D_{c,j}$. 
For the $D_j$-$D_i$ tokens to move from $D_j$ to $D_i$ during the level $x$ opening 
period, the superedge $D_{c,j}$ must open before the superedge $D_{i,b}$, a contradiction. 
Since $D_i$ has exactly two super-out-neighbors at level $x+1$, the same argument applies when $D_j$ has three super-in-neighbors at level $x$. 
We can then conclude that the $D_i$-$D_j$ tokens arrive at a level $x+1$ supernode, in particular the $D_b$ supernode, by the end of the level $x$ opening period, 
and then move from $D_b$ to $D_j$ during the level $x+1$ opening period.

Since the $D_j$-$D_i$ tokens need to be able to move out of $D_b$ and into $D_j$ during the level $x + 1$ opening period, 
both $D_j$ and $D_b$ must have two super-in-neighbors at level $x$, and hence one super-out-neighbor at level $x+2$. 
We denote these neighbors by $D_{j'}$ and $D_{b'}$, respectively. 
To move the $D_i$-$D_j$ tokens to $D_j$ at the opening of $D_{j,j'}$, all $\ell_x$ tokens must be stored in $D_{b'}$. 
However, $D_{b'}$ has at most $3\ell_{x+1} + 1$ vertices in its source, target, and storage 
gadgets, which is insufficient as $\ell_x = Q \cdot \ell_{x+1} > 3\ell_{x+1} + 1$. 
Thus, it is not possible to move the $D_i$-$D_j$ tokens in this case, as needed. 

\noindent{\bf Case 2: $D_b$ has three super-in-neighbors} 
As in the previous case, we first show that it is impossible for both the $D_i$-$D_j$ tokens to reach $D_j$ and the $D_j$-$D_i$ 
tokens to reach $D_i$ during the level $x$ opening period. Recall that when numbering the edges at level $x$ we are free to number the edges 
in any order. Hence, we assume that all the edges incident to a supernode at level $x+1$ must happen consecutively (for each supernode). 
This implies that all edges incident to $D_j$ will either open after or before the $D_i$-$D_j$ tokens have left $D_i$, a contradiction. 

When $D_b$ has three super-in-neighbors we know that $D_b$ has zero neighbors 
at level $x+2$. For the $D_j$-$D_i$ tokens to move from $D_j$ to $D_a$ to $D_b$ and finally to $D_i$ during the level $x$ opening, 
the superedge $D_{i,b}$ must open after the superedge $D_{a,b}$. This implies that by the end of the level $x$ opening period all 
the $D_i$-$D_j$ tokens will be either at supernode $D_b$ or will have moved to other supernodes at level $x+1$. 
However, since $D_b$ has no neighbors at level $x+2$, none of these tokens 
can remain at $D_b$ by the end of the level $x$ opening period. 
Let us denote by $D_1$, $D_2$, up to $D_k$, the at most $k$ supernodes at level $x+1$ that will contain the $D_i$-$D_j$ tokens 
by the end of the level $x$ opening period. Each one of those supernodes, except possibly the last one, will have 
one neighbor at level $x + 2$. Moreover, the number of $D_i$-$D_j$ tokens is $\ell_x = Q \cdot \ell_{x+1}$. 
Finally, we note that all of these tokens have to move to $D_{j'}$ prior to the edge $D_{j,j'}$ opening. 
However, $D_{j'}$ has at most $3\ell_{x+1} + 1$ vertices in its source, target, and storage 
gadgets, which is insufficient as $\ell_x = Q \cdot \ell_{x+1} > 3\ell_{x+1} + 1$. 
Thus, it is again not possible to move the $D_i$-$D_j$ tokens in this case, which completes the proof. 
\end{proof}

Combining Lemmas~\ref{lem-forward-undir} and~\ref{lem-backward-undir} with the fact that $f = O(k)$ and $\ell = k^{O(k)}$, we obtain the following theorem.

\begin{theorem}\label{thm-undir-free}
\textsc{LUTM} is W[1]-hard when parameterized by $\ell + f$. 
\end{theorem}

\bibliography{atoms}

\begin{thebibliography}{10}

\bibitem{DBLP:reference/algo/AlonYZ08}
Noga Alon, Raphael Yuster, and Uri Zwick.
\newblock Color coding.
\newblock In Ming{-}Yang Kao, editor, {\em Encyclopedia of Algorithms - 2008
  Edition}. Springer, 2008.
\newblock \href {https://doi.org/10.1007/978-0-387-30162-4\_76}
  {\path{doi:10.1007/978-0-387-30162-4\_76}}.

\bibitem{DBLP:journals/siamdm/CalinescuDP08}
Gruia C{\u{a}}linescu, Adrian Dumitrescu, and J{\'{a}}nos Pach.
\newblock Reconfigurations in graphs and grids.
\newblock {\em {SIAM} J. Discret. Math.}, 22(1):124--138, 2008.
\newblock \href {https://doi.org/10.1137/060652063}
  {\path{doi:10.1137/060652063}}.

\bibitem{CHJ08}
Luis Cereceda, Jan van~den Heuvel, and Matthew Johnson.
\newblock Connectedness of the graph of vertex-colourings.
\newblock {\em Discrete Mathematics}, 308(56):913--919, 2008.

\bibitem{saurabh-book}
Marek Cygan, Fedor~V. Fomin, Lukasz Kowalik, Daniel Lokshtanov, D{\'{a}}niel
  Marx, Marcin Pilipczuk, Michal Pilipczuk, and Saket Saurabh.
\newblock {\em Parameterized Algorithms}.
\newblock Springer, 2015.

\bibitem{DBLP:books/daglib/0030488}
Reinhard Diestel.
\newblock {\em Graph Theory, 4th Edition}, volume 173 of {\em Graduate texts in
  mathematics}.
\newblock Springer, 2012.

\bibitem{DF97}
Rod~G. Downey and Michael~R. Fellows.
\newblock {\em Parameterized complexity}.
\newblock Springer-Verlag, New York, 1997.

\bibitem{DreyfusWagner}
S.~E. Dreyfus and R.~A. Wagner.
\newblock The {Steiner} problem in graphs.
\newblock 1:195--207, 1972.

\bibitem{flumgrohe}
J{\"{o}}rg Flum and Martin Grohe.
\newblock {\em Parameterized Complexity Theory}.
\newblock Texts in Theoretical Computer Science. An {EATCS} Series. Springer,
  2006.

\bibitem{IDHPSUU11}
Takehiro Ito, Erik~D. Demaine, Nicholas J.~A. Harvey, Christos~H.
  Papadimitriou, Martha Sideri, Ryuhei Uehara, and Yushi Uno.
\newblock On the complexity of reconfiguration problems.
\newblock {\em Theoretical Computer Science}, 412(12-14):1054--1065, 2011.
\newblock URL: \url{http://dx.doi.org/10.1016/j.tcs.2010.12.005}, \href
  {https://doi.org/10.1016/j.tcs.2010.12.005}
  {\path{doi:10.1016/j.tcs.2010.12.005}}.

\bibitem{IKD12}
Takehiro Ito, Marcin Kami\'{n}ski, and Erik~D. Demaine.
\newblock Reconfiguration of list edge-colorings in a graph.
\newblock {\em Discrete Applied Mathematics}, 160(15):2199--2207, 2012.

\bibitem{DBLP:journals/talg/LokshtanovM19}
Daniel Lokshtanov and Amer~E. Mouawad.
\newblock The complexity of independent set reconfiguration on bipartite
  graphs.
\newblock {\em {ACM} Trans. Algorithms}, 15(1):7:1--7:19, 2019.
\newblock \href {https://doi.org/10.1145/3280825} {\path{doi:10.1145/3280825}}.

\bibitem{DBLP:journals/toc/Marx10}
D{\'{a}}niel Marx.
\newblock Can you beat treewidth?
\newblock {\em Theory Comput.}, 6(1):85--112, 2010.
\newblock \href {https://doi.org/10.4086/toc.2010.v006a005}
  {\path{doi:10.4086/toc.2010.v006a005}}.

\bibitem{niedermeier2006}
Rolf Niedermeier.
\newblock {\em Invitation to fixed-parameter algorithms}.
\newblock Oxford Lecture Series in Mathematics and Its Applications. Oxford
  University Press, 2006.

\bibitem{DBLP:journals/algorithms/Nishimura18}
Naomi Nishimura.
\newblock Introduction to reconfiguration.
\newblock {\em Algorithms}, 11(4):52, 2018.
\newblock \href {https://doi.org/10.3390/a11040052}
  {\path{doi:10.3390/a11040052}}.

\bibitem{DBLP:conf/wg/PlehnV90}
J{\"{u}}rgen Plehn and Bernd Voigt.
\newblock Finding minimally weighted subgraphs.
\newblock In Rolf~H. M{\"{o}}hring, editor, {\em Graph-Theoretic Concepts in
  Computer Science, 16rd International Workshop, {WG} '90, Berlin, Germany,
  June 20-22, 1990, Proceedings}, volume 484 of {\em Lecture Notes in Computer
  Science}, pages 18--29. Springer, 1990.
\newblock \href {https://doi.org/10.1007/3-540-53832-1\_28}
  {\path{doi:10.1007/3-540-53832-1\_28}}.

\bibitem{PhysRevA.102.063107}
Kai-Niklas Schymik, Vincent Lienhard, Daniel Barredo, Pascal Scholl, Hannah
  Williams, Antoine Browaeys, and Thierry Lahaye.
\newblock Enhanced atom-by-atom assembly of arbitrary tweezer arrays.
\newblock {\em Phys. Rev. A}, 102:063107, Dec 2020.
\newblock URL: \url{https://link.aps.org/doi/10.1103/PhysRevA.102.063107},
  \href {https://doi.org/10.1103/PhysRevA.102.063107}
  {\path{doi:10.1103/PhysRevA.102.063107}}.

\bibitem{H13}
Jan van~den Heuvel.
\newblock The complexity of change.
\newblock {\em Surveys in Combinatorics 2013}, 409:127--160, 2013.

\end{thebibliography}
\end{document}